\documentclass[11pt]{article}
\pdfoutput=1
\usepackage[a4paper,margin=1in]{geometry}
\usepackage[USenglish]{babel}
\usepackage{xspace}
\usepackage[T1]{fontenc}
\usepackage{amsmath,amssymb,amsthm}
\usepackage{graphicx}
\usepackage{tikz}
\usepackage[table]{xcolor}
\usepackage[bibliography=common]{apxproof}
\usepackage{subcaption}
\usetikzlibrary{arrows.meta}
\usepackage{todonotes}
\usepackage{csquotes}
\usepackage{comment}
\usepackage{url}
\usepackage[hidelinks,pdfusetitle]{hyperref}
\usepackage[capitalise]{cleveref}
\usetikzlibrary{matrix,decorations,arrows,shapes}

\title{Algorithms and Indexing Lower Bounds for Variable String Matching}

\author{
Estéban Gabory\thanks{Supported by the Polish National Science Centre grant number 2023/51/B/ST6/01505.}\\
Institute of Computer Science, University of Wrocław, Wrocław, Poland
}

\date{}

\theoremstyle{plain}
\newtheorem{theorem}{Theorem}
\newtheorem{lemma}[theorem]{Lemma}
\newtheorem{corollary}[theorem]{Corollary}
\newtheorem{conjecture}[theorem]{Conjecture}

\theoremstyle{definition}
\newtheorem{example}[theorem]{Example}

\newcommand{\cO}{\mathcal{O}}
\newcommand{\tcO}{\tilde{\mathcal{O}}}

\newcommand{\T}{\mathcal{T}}

\newcommand{\E}{\mathcal{E}}

\newcommand{\occ}{\textbf{Occ}}

\newcommand{\str}[1]{\mathtt{#1}}

\newcommand{\gd}{\ensuremath{GD}\xspace}

\newcommand{\ed}{\ensuremath{ED}\xspace}

\newcommand{\length}[1]{|#1|}
\newcommand{\size}[1]{||#1||}

\newcommand{\probdef}[3]{
\vspace{2mm}
\noindent\fbox{
   \begin{minipage}{0.96\textwidth}
   \textsc{#1}\\
   {\bf{Input:}} #2  \\
   {\bf{Question:}} #3
   \end{minipage}
   }
   \vspace{2mm}
}

\newcommand{\indexdef}[3]{
\vspace{2mm}
\noindent\fbox{
   \begin{minipage}{0.96\textwidth}
   \textsc{#1}\\
   {\bf{Index:}} #2  \\
   {\bf{Queries:}} #3
   \end{minipage}
   }
   \vspace{2mm}
}
\usepackage[
    backend=biber,
    style=numeric
]{biblatex}

\begin{document}
\maketitle

\begin{abstract}
A \emph{generalized degenerate string} (\gd) is a sequence $T=T_1\cdots T_n$ of nonempty finite sets of strings (called \emph{segments}) such that all strings within each segment have the same length, called the \emph{width} of the segment. We study pattern matching in generalized degenerate strings: given a solid string $P$ and a \gd string $T$, decide whether $P$ occurs in $T$, that is, whether $P$ is a substring of a string obtained by concatenating one string from each segment of $T$. This problem was identified by Ascone et al. (WABI 2024) as the main remaining open case in the fine-grained complexity of pattern matching on \emph{variable strings}, a family of string representations for sets of similar strings with different levels of generality. It is the remaining boundary case between the variants admitting near-linear-time algorithms and those with quadratic conditional lower bounds under SETH. In this paper, we show a $\tcO(N\sqrt{m})$-time algorithm for \gd string matching, where $N$ is the total size of $T$ and $m=|P|$, placing the problem on the subquadratic side of this boundary.

We complement this with upper and lower bounds for the indexed setting. For elastic-degenerate strings (\ed), where the equal-width condition is dropped, Gibney (SPIRE 2020) showed that one can answer pattern matching queries in $\cO(nm^2)$ time where $m$ is the length of the pattern and $n$ is the number of segments of the \ed string, after linear-time preprocessing. We show that this index can be adapted to \gd strings to answer queries in $\cO(nm)$ time. We observe that an existing SETH-based reduction from the same paper already applies to \gd strings and rules out indices with polynomial-time preprocessing and with query time $\cO(n^{1-\varepsilon}m^{\cO(1)}+m)$ for any $\varepsilon>0$. We then prove that, under the $k$-clique conjecture, no combinatorial index with polynomial preprocessing time can answer queries in $\cO(n^{\cO(1)}m^{1-\varepsilon}+m)$ time for any $\varepsilon>0$. For \ed strings, we prove under the same conjecture that no combinatorial index with polynomial preprocessing can answer queries in $\cO(n^{\cO(1)}m^{2-\varepsilon})$ time for any $\varepsilon>0$, matching the quadratic dependence on $m$ in Gibney’s upper bound. Finally, assuming the OMv conjecture, we show that there is no data structure that, after polynomial-time preprocessing of a set of strings and a pattern, answers active-prefix queries on a bit vector of length $m$ in $\cO(m^{2-\varepsilon})$ time for any $\varepsilon>0$; such queries are the standard bottleneck in elastic-degenerate pattern matching. Taken together, these lower bounds suggest that improving the query time for indexed \ed pattern matching below $\cO(n^{\cO(1)}m^2)$ would require both non-combinatorial techniques and an approach that avoids using active-prefix queries as the central bottleneck.
\end{abstract}

\section{Introduction}

\paragraph*{Motivation and Background.} 
String matching (or pattern matching) is a fundamental task in computer science. Given a text and a pattern, one asks whether the pattern occurs in the text, or asks to count or report all occurrences. This simplest setting can be solved in linear time~\cite{crochemoreAlgorithms2007}. This line of research has since been extended to more complex settings, such as approximate matching~\cite{landauefficient1986, landauFast1988, coleDictionary2004, gawrychowskiUnified2018, charalampopoulosFaster2020}, or pattern matching on different data structures~\cite{iliopoulosTruly2016, shiftanSet2016, iliopoulosEfficient2016, grossiOnline2017, aoyamaFaster2018, bernardiniEven2019}. 

A natural generalization is to replace each position of a string by a set (called a \emph{segment}), representing the possible characters at that position. This leads to several data structures with different levels of generality, from \emph{degenerate strings} (or \emph{indeterminate strings}) where the segments are subsets of the alphabet, to \emph{elastic-degenerate strings} where segments can be any finite set of strings (see Figure~\ref{fig:variable_strings}). Ascone et al.~\cite{asconeUnifying2024} studied the fine-grained complexity of pattern matching in each of those variants (collectively called \emph{variable strings}) and showed a clear dichotomy between the cases where a solid pattern can be matched in near-linear time and cases admitting quadratic conditional lower bounds under SETH. More precisely, they gave a near-linear algorithm for pattern matching in a $k$-degenerate string, where each segment is a set of length-$k$ strings, for a fixed integer $k$. For the lower bounds, Backurs and Indyk~\cite{backursWhich2016} showed that a subquadratic algorithm for pattern matching in an elastic-degenerate string would contradict SETH. The intermediate case of \emph{generalized degenerate strings} (\gd), in which the strings within each segment have a common length that may vary between segments, was left as an open boundary case, with an $\cO(N+nm)$-time algorithm, where $N$ is the total size of the text, $n$ is the number of segments, and $m$ is the length of the pattern. 

Generalized degenerate strings were first introduced by Alzamel et al.~\cite{alzamelDegenerate2018a} as a restriction of elastic-degenerate strings. Importantly, they showed that one can solve the intersection problem (namely, assessing whether two given variable strings have a common string in their languages) in linear time for \gd strings. For \ed strings, subquadratic conditional lower bounds and matching upper bounds were shown~\cite{gaboryComparing2023}. It is therefore natural to ask whether the same separation holds for pattern matching. 

Recently, Equi et al.~\cite{equiQuantum2026} gave a subquadratic quantum algorithm for pattern matching in a \gd string, running in $\tcO(\sqrt{mnN})$, but the problem remained open for classical algorithms. In this paper, we answer this question by giving a classical algorithm solving pattern matching in \gd strings in $\tcO(N\sqrt{m})$ time.

Another classical task is \emph{text indexing}, where a text can be preprocessed to answer pattern matching queries faster than the time needed to solve the problem from scratch. In the case of variable strings, a natural goal is to obtain a query time that does not depend on $N$, but instead on the number of segments $n$ and on the pattern length $m$. For \ed strings, the main results were given by Gibney~\cite{gibneyEfficient2020}, who showed the existence of a data structure answering queries in $\cO(nm^2)$ after linear-time preprocessing, and claimed conditional lower bounds under SETH, showing that no index can answer queries in $\cO(n^{\alpha}m^{\beta}+m)$ time with $\alpha<1$ or $\beta<1$. %

\begin{figure}[t]
\centering
\captionsetup[subfigure]{font=small}
\newcommand{\introSet}[1]{\left\{\begin{array}{c}#1\end{array}\right\}}

\begin{subfigure}[t]{0.45\textwidth}
\centering
\(\introSet{\str{A}\\ \str{C}}
 \cdot \introSet{\str{G}}
 \cdot \introSet{\str{T}\\ \str{A}}\)
\caption{}
\label{fig:deg}
\end{subfigure}\hfill
\begin{subfigure}[t]{0.45\textwidth}
\centering
\(\introSet{\str{AC}\\ \str{GT}}
 \cdot \introSet{\str{TG}\\ \str{CA}}
 \cdot \introSet{\str{AA}\\ \str{GC}}\)
\caption{}
\label{fig:kdeg}
\end{subfigure}

\vspace{0.8em}

\begin{subfigure}[t]{0.45\textwidth}
\centering
\(\introSet{\str{AC}\\ \str{GT}}
 \cdot \introSet{\str{A}\\ \str{C}}
 \cdot \introSet{\str{TGA}\\ \str{CAG}}\)
\caption{}\label{fig:gd}
\end{subfigure}\hfill
\begin{subfigure}[t]{0.45\textwidth}
\centering
\(\introSet{\str{AC}\\ \str{G}}
 \cdot \introSet{\varepsilon\\ \str{TGA}}
 \cdot \introSet{\str{C}\\ \str{AGT}}\)
\caption{}
\label{fig:eds}
\end{subfigure}

\caption{Four types of variable strings: \subref{fig:deg}~degenerate strings have width $1$ in every segment; \subref{fig:kdeg}~$k$-degenerate strings have one fixed width $k$; \subref{fig:gd}~\gd strings allow the width to vary across segments; \subref{fig:eds}~\ed strings also allow different string lengths inside a segment.}
\label{fig:variable_strings}
\end{figure}

\paragraph*{Contributions.} We show that pattern matching in a \gd string can be solved in $\tcO(N\sqrt{m})$ time, where $N$ is the total size of the text and $m$ is the length of the pattern. This answers the open question left by Ascone et al.~\cite{asconeUnifying2024}, and places \gd strings on the subquadratic side of the boundary between cases admitting near-linear algorithms and cases with quadratic lower bounds under SETH. The algorithm computes, for each column of the text, the number of segments matched by the pattern when aligned at that column. This array is computed by a heavy-light decomposition of the strings appearing in the segments, according to how often they occur in the pattern. For heavy strings, we use convolution to compute their contribution to the array. For light strings, we use standard pattern matching, but instead of extending candidate matches to the right, as in the usual approach for \ed strings~\cite{bernardiniEven2019}, we scan the strings in the segments and add their contribution directly to the corresponding candidate starting positions.

For indexing, we first observe that one can slightly adapt the index introduced by Gibney~\cite{gibneyEfficient2020} for \ed strings to \gd strings, and obtain an index with linear-time preprocessing and $\cO(nm)$ query time. We then observe that an existing SETH-based reduction from the same paper already applies to \gd strings, showing that there is no index with polynomial preprocessing that can answer queries in $\cO(n^{\alpha}m^{\beta}+m)$ time with $\alpha<1$. However, we notice a potential flaw in the lower bound forbidding $\beta<1$ in the same paper, which seems to suggest that the question of whether one can obtain an index for \ed strings with $\cO(n^{\cO(1)}m^{1-\varepsilon}+m)$ query time after polynomial-time preprocessing for some $\varepsilon>0$ is still open. To partially address this gap, we give two alternative lower bounds, showing that, under the $k$-clique conjecture, no \emph{combinatorial} index with polynomial preprocessing can answer queries in $\cO(n^{\alpha}m^{\beta}+m)$ time with $\beta<1$ for \gd strings, and $\beta<2$ for \ed strings. We then consider the \emph{active-prefix} problem, which is often considered the main bottleneck in \ed string indexing, and asks, given a set $S$ of strings of total size $N$, a pattern $P$ and a bit vector $U$ both of length $m$, to compute all positions $j$ such that $P[i+1\dots j]\in S$ and such that $U[i]=1$ for some $i \le j$. We show that, under the OMv conjecture~\cite{henzingerUnifying2015}, one cannot index a set and a pattern in polynomial time to answer active-prefix queries with input bit vectors of length $m$ in $\cO(m^{2-\varepsilon})$ time for any $\varepsilon>0$. The proof is based on a reduction from~\cite{bernardiniEven2019}. Taken together, these lower bounds suggest that improving the query time for indexed \ed pattern matching below $\cO(n^{\cO(1)}m^2)$ would require both non-combinatorial techniques and an approach that avoids using active-prefix queries as the central bottleneck.

\section{Preliminaries}
\subsection{Strings and elastic-degenerate strings}

A \emph{string} (or \emph{solid string}) $S$ over an alphabet $\Sigma$ is a finite sequence of characters from $\Sigma$. The \emph{length} of $S$ is denoted by $|S|$, and the \emph{empty string} is denoted by $\varepsilon$. We write $S[i\dots j]$ for the \emph{substring} of a string $S$ between positions $i$ and $j$. In this case, we say that the substring $S[i\dots j]$ \emph{occurs} in $S$.

An \emph{elastic-degenerate (\ed) string} $T$ over an alphabet $\Sigma$ is a sequence $T=T_{1} \cdots T_{n}$ of $n$ nonempty finite sets, called \emph{segments}, where each $T_{i}$ is a subset of $\Sigma^*$. We denote the \emph{length} of $T$ by $\length{T}=n$, and the \emph{size} $\size{T}=N$ is the total number of characters in $T$, that is,
$$N=N_{\varepsilon}+\sum^{n}_{i=1}\sum_{S\in T_{i}} |S|,$$
where $N_{\varepsilon}$ is the total number of empty strings in the segments of $T$. The \emph{cardinality} $B$ of $T$ is defined as $B=\sum_{i=1}^n |T_{i}|$. For any $1\leq i\leq j\leq n$, $T[i\dots j]$ denotes the \ed string $T_{i}\cdots T_{j}$, which is the fragment between segments $i$ and $j$ of $T$. Given two sets of strings $S$ and $S'$, we define the \emph{Cartesian product} of $S$ and $S'$ as the set $S \times S' = \{ss' : s\in S, s'\in S'\}$ of all concatenations of a string in $S$ with a string in $S'$. Given an \ed string $T=T_{1}\cdots T_{n}$, we define its \emph{language} $\mathcal{L}(T)$ as the Cartesian product of its segments, that is $\mathcal{L}(T)=T_{1}\times \cdots \times T_{n}$. We say that a string $P$ \emph{occurs} in $T$ if it occurs in a string $S\in\mathcal{L}(T)$. Given a set of strings $S$ and an integer $k\ge 1$, we denote by $S^k$ the \ed string of length $k$ whose segments are all equal to $S$.

If for every $i$ the strings in $T_{i}$ all have the same length $k_i\ge 1$ (called the \emph{width} of $T_{i}$), we say that $T$ is a \emph{generalized degenerate (\gd) string}. If in addition all segments $T_{i}$ have the same width $k$, $T$ is a \emph{$k$-degenerate} string ($k$-D, in short). In the special case $k=1$, $T$ is known in the literature as a \emph{degenerate} or \emph{indeterminate} string. 

Given a \gd string $T$, the \emph{span} $L$ of $T$ is defined as $L=\sum_{i=1}^n k_i$, where $k_i$ is the width of segment $T_{i}$. Note that every string $s\in \mathcal{L}(T)$ has length $L$. For each segment $T_i$, we define $L_i=\sum_{j=1}^{i-1} k_j+1$ and $R_i=L_i+k_i-1$, that is, $L_i$ and $R_i$ are the leftmost and rightmost positions of $T_i$ within $T$, respectively. 
Let us fix a solid string $P$ of length $m$. We say that $P$ \emph{occurs} in $T$ at \emph{position} (or \emph{column}) $1\le p\le L-m+1$ if there exists a string $S\in \mathcal{L}(T)$ such that $P$ is a substring of $S$ starting at position $p$.  We write $s(p)$ for the index of the segment containing position $p$, namely $s(p)=i$ if and only if $L_i \le p \le R_i$, and $t(p)=s(p+m-1)$. In particular, if $P$ occurs at position $p$ in $T$, then $P$ occurs in the \gd string $T[s(p) \dots t(p)]$, at position $p-L_{s(p)}+1$.

Given a width $k$, we write $N_k=\sum_{i,~k_i=k} ||T_i||$ for the total size of the strings of width $k$ that occur in $T$, and $B_k=\sum_{i: k_i=k} |T_i|$ for their total cardinality. We write $\mathcal{V}_k$ for the set of strings occurring in a segment of width $k$ in $T$. Note that $|\mathcal{V}_k|\le B_k$, since $B_k$ counts repetitions, and that $\sum_k N_k=N$ and $\sum_k B_k=B$. 

\begin{example}
Consider the \gd string $T=\left\{\begin{array}{c}ba\\ab\end{array}\right\}\cdot\left\{\begin{array}{c}b\\a\end{array}\right\}\cdot\left\{\begin{array}{c}aa\\ba\end{array}\right\}\cdot\left\{\begin{array}{c}a\end{array}\right\}$. It has length $n=4$, size $N=11$, and cardinality $B=7$. We have $k_1=k_3=2$, and $k_2=k_4=1$. The span of $T$ is $L=k_1+k_2+k_3+k_4=6$. Furthermore, $N_1=3$, $N_2=8$, $B_1=3$, and $B_2=4$, and $\mathcal{V}_1=\{a,b\}$, $\mathcal{V}_2=\{ab,ba,aa\}$. The string $P=aba$ occurs in $T$ at positions $1$, $2$, and $3$.
\end{example}

\section{Pattern matching}

We want to solve the following problem:

\probdef{GDSM}{A \gd string $T$, a solid string $P$}{Does $P$ occur in $T$?}

We will compute, for each candidate starting position $1\le p \le L-m+1$, the number of segments that are correctly matched by the alignment of $P$ starting at $p$. If this number is equal to the number of segments spanned by the alignment, then $P$ occurs at position $p$. In the following lemma, we show how to compute the number of matched segments that are \emph{fully covered} by the alignment, for a fixed width $k$. The other cases are handled separately. 

\begin{figure}[t]
\centering
\includegraphics[width=.9\textwidth]{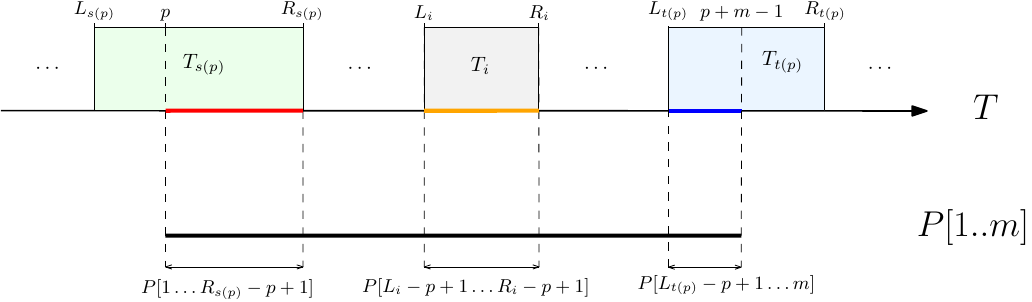}
\caption{Alignment of a candidate occurrence of $P[1\dots m]$ starting at position $p$ with the segments of $T$. If $T_i$ has width $k$, then it contributes $1$ to $A[p]$ if there exists a string $t \in T_i$ such that $t = P[L_i-p+1 \dots R_i-p+1]$. After processing all widths, $A[p]$ equals the number of fully covered segments matched by this alignment.}
\label{fig:top-alignment}
\end{figure}

\begin{lemma}\label{lem:anchor}
Let $T=T_1\dots T_n$ be a \gd string, let $P$ be a solid string, and let $\T_P$ be the suffix tree of $P$, annotated with the number of occurrences of every represented substring. Fix a width $k$ occurring in $T$, and let $A$ be an array of length $L-m+1$. In $\cO\!\left(N_k+\sqrt{LmB_k\log m}\right)$ time, one can add to every $A[p]$ the number of segments $T_i$ of width $k$ such that $p\le L_i$, $p+m-1\ge R_i$, and $P[L_i-p+1\dots R_i-p+1]\in T_i$ (see Figure~\ref{fig:top-alignment}).
\end{lemma}

\begin{proof}
We consider only strings in $\mathcal{V}_k$. For every $t\in \mathcal{V}_k$, we say that $t$ is \emph{heavy} if $t$ occurs at least $\tau_k$ times in $P$, where $\tau_k$ is a parameter to be fixed later. Otherwise, $t$ is \emph{light}. Using the suffix tree $\T_P$ given as input, classifying the strings as heavy or light takes $\cO(N_k)$ time.

For each heavy string $s$ of length $k$, we construct a binary string $T_s$ of length $L$ by setting $T_s[p]=1$ if $p=L_i$ for some segment $T_i$ of width $k$ containing $s$, and $T_s[p]=0$ otherwise. We also construct a binary string $P_s$ of length $m$ by setting $P_s[j]=1$ if an occurrence of $s$ starts at position $j$ in $P$, and $P_s[j]=0$ otherwise. We then compute, for every starting position $p\in[1,L-m+1]$, the convolution
\[
H_s[p]=\sum_{j=1}^{m-k+1} P_s[j]\cdot T_s[p+j-1]
\]
in $\cO(L\log m)$ time per heavy string via FFT, and add $H_s[p]$ directly to $A[p]$. Since there are at most $m/\tau_k$ heavy strings of length $k$, this takes $\cO((mL/\tau_k)\log m)$ time.

To count the matches with light strings, we read every length-$k$ light string from $T$ in $\T_P$ and add $1$ directly to $A[L_i-j+1]$ for every occurrence of that string starting at position $j$ in $P$ and every width-$k$ segment $T_i$ containing it (see Figure~\ref{fig:light-strings}). This takes $\cO(N_k+B_k\tau_k)$ time, since we read at most $B_k$ strings, and each light string occurs fewer than $\tau_k$ times in $P$.

The total running time is therefore $\cO\left(N_k+\frac{mL}{\tau_k}\log m+B_k\tau_k\right)$. Setting $\tau_k=\sqrt{Lm\log m/B_k}$, we obtain $\cO\left(N_k+\sqrt{LmB_k\log m}\right)$.
\end{proof}

\begin{figure}
\centering
\includegraphics[width=.9\textwidth]{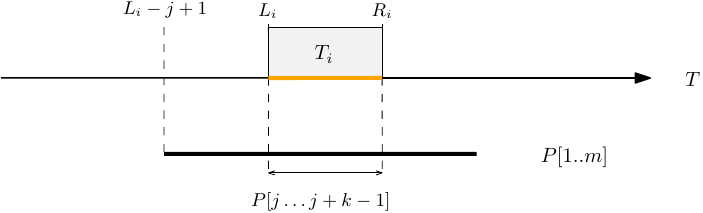}
\caption{For the light strings, when reading a string $t$ in $T_i$ that matches a substring of $P$ starting at position $j$, we add $1$ to $A[L_{i}-j+1]$.}
\label{fig:light-strings}
\end{figure}

\begin{theorem}\label{main}
We can solve GDSM in $\cO(N+\sqrt{NLm\log L\log m})=\tcO(N\sqrt{m})$ time.
\end{theorem}

\begin{proof}
If $m>L$, there is no occurrence, so assume $m\le L$. In line with previous work on pattern matching in variable strings~\cite{grossiOnline2017,aoyamaFaster2018,bernardiniEven2019,asconeUnifying2024}, we consider the following four subproblems: (i) If $k_i\ge m$, does $P$ occur in a string from $T_{i}$? (\emph{easy case}) (ii) Compute every suffix of a string $t\in T_{i}$ that matches a prefix of $P$ (\emph{suffix case}); (iii) Compute every prefix of a string $t\in T_{i}$ that matches a suffix of $P$ (\emph{prefix case}); (iv) Find the strings $t\in T_{i}$ that are substrings of $P$ (\emph{anchor case}). We start by constructing $\T_P$, the suffix tree of $P$, in $\cO(m)$ time.

We construct bit vectors $V_{\mathrm{easy}}$, $V_{\mathrm{pref}}$, $V_{\mathrm{suf}}$, and $V_{\mathrm{anchor}}$ of length $L-m+1$, where each bit vector marks the starting positions that are compatible with the corresponding case. Intuitively, $V_{\mathrm{easy}}$ marks the occurrences fully contained in one segment, $V_{\mathrm{suf}}$ checks the left boundary segment of an occurrence spanning at least two segments, $V_{\mathrm{pref}}$ checks the right boundary segment, and $V_{\mathrm{anchor}}$ checks all segments fully contained in the occurrence. More formally (see Figure~\ref{fig:top-alignment}), let $s(p)$ be the index of the segment containing position $p$, and let $t(p)=s(p+m-1)$. We define $V_{\mathrm{easy}}[p]=1$ if and only if there exists a string in $T_{s(p)}$ in which $P$ occurs at position $p-L_{s(p)}+1$. Next, $V_{\mathrm{suf}}[p]=1$ if and only if $s(p)<t(p)$ and there exists a string $t\in T_{s(p)}$ having
$P[1 \dots R_{s(p)}-p+1]$ as a suffix.
Similarly, $V_{\mathrm{pref}}[p]=1$ if and only if $s(p)<t(p)$ and there exists a string $t\in T_{t(p)}$ having
$P[L_{t(p)}-p+1 \dots m]$ as a prefix.
Finally, $V_{\mathrm{anchor}}[p]=1$ if and only if, for every segment $i$ such that $p\le L_{i}$ and $p+m-1\ge R_i$, there exists a string $t\in T_i$ satisfying
$t=P[L_{i}-p+1 \dots R_i-p+1]$.
Then the set of occurrences of $P$ in $T$ is exactly the set of positions $p$ such that $V_{\mathrm{easy}}[p]=1$ or $(V_{\mathrm{suf}}[p]\wedge V_{\mathrm{pref}}[p]\wedge V_{\mathrm{anchor}}[p])=1$. It is straightforward to solve the easy case and the suffix/prefix case in $\cO(N+m)$ time using any standard pattern-matching algorithm (e.g. KMP~\cite{knuthFast1977}) and using $\T_P$, respectively (see~\cite{aoyamaFaster2018} for details). We now focus on computing $V_{\mathrm{anchor}}$.

To compute $V_{\mathrm{anchor}}$, we count, for each starting position $p$, the number of fully covered segments whose matching condition is satisfied, and compare it with the total number of segments fully covered by the corresponding alignment of $P$. More precisely, let $G[p]$ denote the number of segments $T_i$ such that $p\le L_i$ and $p+m-1\ge R_i$. The array $G$ has length $L-m+1$ and can be computed in $\cO(L)$ time using prefix sums. We initialize an array $A$ of length $L-m+1$ to zero and, for every width $k$ occurring in $T$, apply Lemma~\ref{lem:anchor} to add the contribution of the segments of width $k$ to $A$. After all widths have been processed, $A[p]$ is the number of fully covered segments whose matching condition is satisfied. Hence, $V_{\mathrm{anchor}}[p]=1$ if and only if $A[p]=G[p]$. By Lemma~\ref{lem:anchor}, processing all widths takes $\cO\left(\sum_k\left(N_k+\sqrt{LmB_k\log m}\right)\right)$ time. Over all $k$, using Cauchy--Schwarz and the fact that $\sum_k k B_k= \sum_k N_k = N$, we observe that $$\sum_k{\sqrt{B_k}}=\sum_k\sqrt{{\frac{1}{k}kB_k}}=\sum_k\sqrt{{\frac{1}{k}N_k}}\le \sqrt{\sum_k{\frac{1}{k}}\sum_k{N_k}}=\cO(\sqrt{N\log L}).$$ Wrapping up, the total time complexity over all widths is
\begin{align*}
\cO\left(L+m+\sum_k\left(N_k+\sqrt{LmB_k\log m }\right)\right)&=\cO\left(N+\sqrt{Lm\log m}\sum_k\sqrt{B_k}\right)\\
&=\cO\left(N+\sqrt{NLm\log L\log m}\right),
\end{align*}

Using $L \le N$, this is $\tcO(N\sqrt{m})$.
\end{proof}

\section{Indexing}

\subsection{On previous results}

For \ed strings, Gibney~\cite{gibneyEfficient2020} gave the main indexing results, and several of them can be directly applied in our setting. In particular, an index for \ed strings was given, answering queries in $\cO(nm^2)$ time after linear-time preprocessing. In fact, by slightly adapting the index, we can obtain an index for \gd strings with linear-time preprocessing and $\cO(nm)$ query time.  

\begin{theorem}\label{thm:index}
A \gd string $T$ can be preprocessed in $\cO(N)$ time so that pattern matching queries are answered in $\cO(nm)$ time.
\end{theorem}
\begin{proof}
For every segment $T_i$, we build an Aho--Corasick automaton for the set of strings in $T_i$. We also build two generalized suffix trees: one for $T_i$, and one for the set obtained by reversing every string in $T_i$. The total size of these data structures over all segments is $\cO(N)$.

When a query pattern $P$ of length $m$ is given, we process the segments from left to right, maintaining a bit vector that marks the prefixes of $P$ matched by the segments processed so far. As in \cref{main}, we distinguish the four cases: easy, suffix, prefix, and anchor. The easy case is handled by searching $P$ in the generalized suffix tree of $T_i$, which takes $\cO(m)$ time. The suffix and prefix cases are handled by searching for the prefixes and suffixes of $P$ that match strings in $T_i$, which can be done in $\cO(m)$ time using the generalized suffix trees. Finally, the anchor case is handled by reporting the occurrences of strings of $T_i$ as substrings of $P$. This takes $\cO(m+\occ)=\cO(m)$ time using the Aho--Corasick automaton, since all strings in $T_i$ have the same length and therefore the cumulative number of occurrences is at most $m$. Since each segment is processed in $\cO(m)$ time, the total query time is $\cO(nm)$.
\end{proof}

Another result from~\cite{gibneyEfficient2020} is a conditional lower bound for indexing \ed strings, showing that no index with polynomial preprocessing can answer queries in $\cO(n^{\alpha}m^{\cO(1)})$ time with $\alpha<1$, unless SETH is false. The proof is based on a reduction from the Orthogonal Vectors problem:

\probdef{Orthogonal Vectors (OV)}{Two sets $X$ and $Y$ of $n$ vectors in $\{0,1\}^d$, where $d=\Theta(\log n)$}{Does there exist a pair $(a,b)\in X\times Y$ such that $a\cdot b=0$?}

The \emph{Orthogonal vectors hypothesis} (OVH) states that OV cannot be solved in $\cO(n^{2-\varepsilon})$ time for any $\varepsilon>0$. It is known that OVH is implied by SETH~\cite{williamsnew2005}. Importantly, OVH also implies hardness for the indexing version of the problem:

\begin{lemma}[\cite{equiGraphs2020,equiGraphs2023}]
   Under OVH, one cannot solve OV in $\cO(|X|^{\alpha}|Y|^{\beta})$ with $\alpha<1$ or $\beta<1$, even after allowing arbitrary polynomial-time preprocessing of $X$.
\end{lemma}

We observe that the reduction of Gibney actually uses generalized degenerate strings, hence the same lower bound applies to \gd strings:

\begin{theorem}[\cite{gibneyEfficient2020}]
Over an alphabet of size two, there is no solution for \gd-indexing with query time $\cO(n^\alpha m^{\cO(1)} + m)$ for constant $\alpha < 1$ under OVH (and hence SETH), even with arbitrary polynomial preprocessing.
\end{theorem}

Finally, another reduction from the same paper gives the following result:

\begin{theorem}[\cite{gibneyEfficient2020}]
OV-indexing on vector sets $X,Y\subseteq\{0,1\}^d$ can be reduced to \ed-indexing over an \ed-text $T$, constructed from $X$, of length $n=\cO(d(|X|+|Y|))$, of size $N=\cO(d(|X|^2+|Y|))$, and having an alphabet of size four. One pattern $P$ needs to be given as a query, where $P$ is constructed from $Y$ and has length $m=\cO(d|Y||X|)$.
\end{theorem}

It is then claimed that this implies a lower bound that prohibits any index with polynomial preprocessing that can answer queries in $\cO(n^{\alpha}m^{\beta}+m)$ time with $\beta<1$, unless SETH is false. However, the proof appears not to account for the additive $+m$ term, which is necessary to read the pattern $P$ in the query. In fact, since $m=\cO(d|Y||X|)$, simply reading the pattern already takes enough time to solve the OV instance without contradicting OVH. We did not find a way to fix the reduction, and leave the existence of such a lower bound as an open problem. However, in what follows, we show several hardness results for indexing \gd and \ed strings, which suggest that the claimed lower bound for \ed strings is likely to hold.

\subsection{Combinatorial lower bounds}

In this section we prove that the $\cO(nm)$-time queries for \gd pattern matching from \cref{thm:index}, and the $\cO(nm^2)$ for \ed pattern matching from~\cite{gibneyEfficient2020}, are essentially optimal for combinatorial algorithms under the $k$-clique conjecture.

\probdef{$k$-clique}{A graph $G$ with $n$ vertices}{Does $G$ contain a clique of size $k$?}

\begin{conjecture}
No combinatorial algorithm solves $k$-clique on $n$ vertices in time $\cO(n^{k-\varepsilon})$ for any $\varepsilon>0$.
\end{conjecture}

We adapt the reduction from Bringmann et al.~\cite{bringmannDichotomy2016}, which was originally designed for the Word Break problem. For a graph $G$ with $V(G)=\{v_1,\dots,v_n\}$, set $\Sigma=\{1,\dots,n,\#\}$ where each integer is a distinct symbol. We define the string $M = 1\,2\,\dots\,n$.
For each vertex $v_i$ we set
$$
s(v_i)=\#^{n-i+1}\,M[1\dots i-1],\qquad
t(v_i)=M[i\dots n]\,\#^{i-1},
$$
both of length $n$. For each edge $(v_i,v_j)\in E(G)$ (since the graph is undirected, we represent every edge by both orientations) we write
$$
e(v_i,v_j)=M[i\dots n]\,\#\,j\,\#\,M[1\dots i-1]
$$
of length $n+3$. Then set
$$
S(G)=\{s(v):v\in V(G)\},\;
T(G)=\{t(v):v\in V(G)\},\;
\E(G)=\{e(u,v):(u,v)\in E(G)\}.
$$

Let $T^k_G$ be the following \gd string:
$$
T^k_G = S(G)\;\cdot\;\E(G)^{k-1}\;\cdot\;T(G).
$$
Assuming that $k$ is constant, $|T^k_G| = \cO(1)$ and $\|T^k_G\| = \cO(n^3)$.

For a set $S=\{v_{i_1},\dots,v_{i_{k-1}}\}$ of $k-1$ vertices, we define the pattern
$$
P_S = M \,\#\, i_1 \,\#\, M \,\#\, i_2 \,\#\, \cdots \,\#\, i_{k-1} \,\#\, M,
$$
so $|P_S| = \Theta(n)$.

\begin{lemma}[\cite{bringmannDichotomy2016}]\label{lem:gd-reduction}
Let $k$ be a fixed integer, and $G$ be a graph with $n$ vertices. Let $S=\{v_{i_1},\dots,v_{i_{k-1}}\}$ be a $(k-1)$-clique in $G$. Then $P_S$ occurs in $T^k_G$ if and only if there exists a vertex $u$ such that $S\cup\{u\}$ is a $k$-clique in $G$.
\end{lemma}

\begin{corollary}\label{cor:reduction}
Let $k$ be a fixed integer, and $G$ be a graph with $n$ vertices. We can build a \gd string $T^k_G$ of length $\cO(1)$ and size $\cO(n^3)$ and, for any $(k-1)$-clique $S$, a pattern $P_S$ of length $\Theta(n)$ such that $P_S$ occurs in $T^k_G$ if and only if $S$ can be extended to a $k$-clique.
\end{corollary}

\begin{theorem}
Unless the $k$-clique conjecture fails, there is no combinatorial algorithm for \gd indexing with polynomial preprocessing and query time $\cO(n^{\cO(1)} m^{1-\varepsilon}+m)$ for any $\varepsilon>0$.
\end{theorem}
\begin{proof}
Assume such a combinatorial algorithm exists with preprocessing $\cO(N^b)$ and query $\cO(n^{d} m^{1-\varepsilon}+m)$ for constants $b,d$ and $\varepsilon>0$. We show how to solve $k$-clique in $\cO(n^{k-\varepsilon})$ time, contradicting the conjecture.

For each vertex $v$, let $G_v$ be the subgraph induced by the neighbors of $v$ (excluding $v$ itself). $G$ has a $k$-clique if and only if some $G_v$ has a $(k-1)$-clique.  
We construct the \gd string $T^{k-1}_{G_v}$ (using the construction above with $k-1$ instead of $k$) for every $v$. We then concatenate these strings in groups of $\lfloor n^{\gamma}\rfloor$ (for some $\gamma>0$ to be determined later), separated by a fresh delimiter $\$$, obtaining $n^{1-\gamma}$ \gd strings each of size $\cO(n^{3+\gamma})$ and length $\cO(n^{\gamma})$. We preprocess each in time $\cO((n^{3+\gamma})^b)=\cO(n^{3b+b\gamma})$, so the total preprocessing time is $\cO(n^{3b+1+\gamma(b-1)})$.

Now take any $(k-2)$-clique $S$ in $G$. We build the pattern $P_S$ as in Corollary~\ref{cor:reduction} (with $k-1$ in place of $k$); its length is $\Theta(n)$. We can now query every preprocessed \gd string with the pattern $P_S$. By Corollary~\ref{cor:reduction}, a query returns YES if and only if there exist vertices $u,v$ such that $S\cup\{u,v\}$ is a $k$-clique in $G$. Each query costs $\cO((n^{\gamma})^{d} n^{1-\varepsilon}+n)=\cO(n^{\gamma d+1-\varepsilon}+n)$.

There are at most $n^{k-2}$ choices for $S$, and each is tested against $n^{1-\gamma}$ indices. Hence, the total time to solve $k$-clique is
$$
\cO\!\left(n^{3b+1+\gamma(b-1)} + n^{k-2}n^{1-\gamma}\bigl(n^{\gamma d+1-\varepsilon}+n\bigr)\right)
= \cO\!\left(n^{3b+1+\gamma(b-1)} + n^{k-\varepsilon+\gamma(d-1)} + n^{k-\gamma}\right).
$$
Set $\gamma = \varepsilon/d$. Then the second and the third terms become $n^{k-\varepsilon/d}$. Because $k$ is constant and $b,d$ are fixed, taking $k$ large enough makes the first term also $\cO(n^{k-\varepsilon})$. Hence we would solve $k$-clique with a combinatorial algorithm in $\cO(n^{k-\varepsilon'})$ time for some $\varepsilon'>0$, contradicting the conjecture.  
\end{proof}

For \ed strings, we use almost the same gadgets as above, but the fact that strings within sets can have different lengths allows us to directly build an index that extends a $(k-2)$-clique with two vertices in a single pattern matching query. The main difference is an ED connector that lets the pattern choose one extension vertex on the left of the connector and another on the right. Add two fresh symbols $\$$ and $?$ to the alphabet, and keep $M$, $e(\cdot,\cdot)$, and $\E(G)$ as above. For each vertex $v_i$ and each edge $(v_i,v_j)$, we define
$$
\tilde{s}(v_i)=\$\,M[1\dots i-1],\qquad
\tilde{t}(v_i)=M[i\dots n]\,\$,\qquad
g(v_i,v_j)=M[i\dots n]\,\$\,?\,\$\,M[1\dots j-1].
$$
Let
$$
\begin{aligned}
\tilde{S}(G)&=\{\tilde{s}(v):v\in V(G)\},\\
\tilde{T}(G)&=\{\tilde{t}(v):v\in V(G)\},\\
\Gamma(G)&=\{g(u,v):(u,v)\in E(G)\}.
\end{aligned}
$$
We construct the ED string
$$
\tilde{T}^{k}_G
=\tilde{S}(G)\cdot \E(G)^{k-2}\cdot \Gamma(G)\cdot \E(G)^{k-2}\cdot \tilde{T}(G).
$$
It has length $\cO(1)$ and size $\cO(n^3)$.

For a $(k-2)$-clique $S=\{v_{i_1},\dots,v_{i_{k-2}}\}$, define
$$
\tilde{P}_S
=\$\,M\,\#\,i_1\,\#\,M\,\#\,i_2\,\#\,\cdots\,\#\,i_{k-2}\,\#\,M\,
\$\,?\,\$\,M\,\#\,i_1\,\#\,M\,\#\,i_2\,\#\,\cdots\,\#\,i_{k-2}\,\#\,M\,\$.
$$
Clearly, $|\tilde{P}_S|=\cO(n)$.

\begin{lemma}[\cite{bringmannDichotomy2016}]\label{lem:ed-reduction}
The pattern $\tilde{P}_S$ occurs in $\tilde{T}^{k}_G$ if and only if there exist two vertices $u,v$ such that $S\cup\{u,v\}$ is a $k$-clique in $G$.
\end{lemma}

\begin{corollary}\label{cor:ed-reduction}
Let $k$ be a fixed integer. Given a graph $G$ on $n$ vertices, we can construct an ED string $\tilde{T}^{k}_G$ of length $\cO(1)$ and size $\cO(n^3)$ such that for any $(k-2)$-clique $S$ of $G$, we can construct a pattern $\tilde{P}_S$ of length $\cO(n)$ such that $\tilde{P}_S$ occurs in $\tilde{T}^{k}_G$ if and only if there exist two vertices $u,v$ such that $S\cup \{u,v\}$ is a $k$-clique.
\end{corollary}

\begin{theorem}
There is no combinatorial algorithm for \ed string indexing that answers queries in $\cO(n^{\cO(1)}m^{2-\varepsilon})$ time after polynomial preprocessing, unless the $k$-clique conjecture fails.
\end{theorem}
\begin{proof}
Assume a combinatorial algorithm for \ed string indexing with preprocessing time $\cO(N^b)$ and query time $\cO(n^{d}m^{2-\varepsilon})$ for fixed constants $b, d$, and $\varepsilon>0$.

We construct $\tilde{T}^{k}_G$ of length $\cO(1)$ and size $\cO(n^3)$ as in Corollary~\ref{cor:ed-reduction}, and preprocess it in $\cO(n^{3b})$ time to answer pattern matching queries for patterns of length $m$ in $\cO(n^{2-\varepsilon})$ time. Given a $(k-2)$-clique $S$, we construct the pattern $\tilde{P}_S$ as in Corollary~\ref{cor:ed-reduction} and query the \ed string index with $\tilde{P}_S$. By Corollary~\ref{cor:ed-reduction}, $\tilde{P}_S$ has length $\cO(n)$, and the query returns YES if and only if there exist two vertices $u,v$ such that $S\cup \{u,v\}$ is a $k$-clique. Each query takes $\cO(n^{2-\varepsilon})$ time. Since we need at most $\cO(n^{k-2})$ queries to solve $k$-clique, the total time is
$$\cO(n^{3b} + n^{k-2} n^{2-\varepsilon}) = \cO(n^{3b} +  n^{k-\varepsilon}).$$
For sufficiently large constant $k>3b+\varepsilon$, this is $\cO(n^{k-\varepsilon})$, contradicting the $k$-clique conjecture.
\end{proof}

\subsection{Active-prefix problem}

Let us consider the following problem:

\probdef{Active-prefix (AP)}{A string $P$ of length $m$, a bit vector $U$ of length $m$, and a set $S$ of strings of total size $N$}{Compute a bit vector $V$ of length $m$ such that $V[j]=1$ if and only if there exist $i\le j$ and $s\in S$ such that $U[i]=1$ and $P[1\dots i]\cdot s=P[1\dots j]$.}

This problem is often considered the main bottleneck for elastic-degenerate string matching, as every existing algorithm computes partial matches from left to right and extends them upon processing the segments of the elastic-degenerate text~\cite{grossiOnline2017,aoyamaFaster2018,bernardiniEven2019,asconeUnifying2024}. We consider the Online Boolean Matrix-Vector multiplication (OMv) conjecture~\cite{henzingerUnifying2015}, and show that active-prefix queries cannot be answered in $\cO(m^{2-\varepsilon})$ time for any $\varepsilon>0$ after polynomial-time preprocessing of the string set and the pattern, unless the OMv conjecture is false. This suggests that a query time of $\cO(n^{\cO(1)}m^{2-\varepsilon})$ for some $\varepsilon>0$ is unlikely to be achieved for elastic-degenerate string indexing, unless this comes from a radically different approach than the existing ones. In particular, the $\cO(nm^2)$-time index from~\cite{gibneyEfficient2020} is likely to be optimal, unless one can improve upon it with a non-combinatorial algorithm without relying on the AP problem.

We consider the following problem:

\indexdef{Online Boolean Matrix-Vector multiplication (OMv)}{A Boolean matrix $M$ of size $n\times n$}{Given a sequence of vectors $v_1\dots v_{n}$, after receiving $v_i$, return $M v_i$ before receiving $v_{i+1}$.}

\begin{conjecture}[\cite{henzingerUnifying2015}]
For any constant $\varepsilon >0$, there is no algorithm for OMv that answers all queries in $\cO(n^{3-\varepsilon})$ time, even after arbitrary polynomial-time preprocessing of $M$.
\end{conjecture}

The reduction is based on that of Bernardini et al.~\cite{bernardiniEven2019}, which reduces Boolean matrix multiplication to AP and gives a combinatorial lower bound for the latter problem. We observe that the same construction can be interpreted in the indexed setting: after preprocessing the set of strings encoding the matrix, each active-prefix query corresponds to an online vector query. This yields a lower bound for indexed AP under the OMv conjecture.

\begin{theorem}\label{thm:ap-omv}
Assuming the OMv conjecture, one cannot preprocess a set of strings $S$ of total size $N$ and a pattern $P$ of length $m$ in polynomial time in $N+m$ and answer active-prefix queries for an input vector $U$ in $\cO(m^{2-\varepsilon})$ time for any $\varepsilon>0$.
\end{theorem}
\begin{proof}

Let $M$ be an $n\times n$ Boolean matrix.

Let the alphabet be $\{a,b\}$ and set $P=a^n b a^n$. For every entry $M[i,j]=1$, add the string $s_{i,j}=a^{n-j}ba^i$ to $S$. Then $|P|=2n+1=\Theta(n)$, and $||S||\le n^2(2n+1)=\cO(n^3)$. Thus $S$ can be constructed and preprocessed in polynomial time in the
size of $M$.

Now consider an online vector $v\in\{0,1\}^n$. We build the AP bit vector $U=v\cdot 0^{n+1}$, so $U[k]=v[k]$ for $1\le k\le n$, and $U[k]=0$ for $n<k\le 2n+1$. Let $V$ be the answer to the AP query $(P,U)$.

We claim that for every $i\in\{1,\dots,n\}$,
$$
V[n+1+i]=1
\quad\Longleftrightarrow\quad
(Mv)[i]=1 .
$$
Indeed, for $k,i,j\in\{1,\dots,n\}$,
$$
P[1\dots k]\cdot s_{i,j}
=a^k a^{n-j}ba^i
=a^{k+n-j}ba^i .
$$
This is a prefix of $P=a^n b a^n$ if and only if $k=j$, and then it is
exactly $P[1\dots n+1+i]$. Therefore $V[n+1+i]=1$ if and only if there
exists $j$ such that $U[j]=1$ and $M[i,j]=1$, which is precisely
$(Mv)[i]=1$.

Consequently, one AP query computes the matrix-vector product for one
online vector. If AP queries could be answered in
$\cO(m^{2-\varepsilon})$ time after polynomial preprocessing of $S$ and $P$,
then any sequence of online vectors could be processed by issuing one AP
query per vector. In particular, the total query time would be $\cO(n^{3-\varepsilon})$. This
contradicts the OMv conjecture.
\end{proof}

\section{Conclusion and future work}

We studied the problem of pattern matching in generalized degenerate strings, and we gave an algorithm running in $\tcO(N\sqrt{m})$ time, which is the first classical subquadratic-time algorithm for this problem. We also gave a combinatorial lower bound for indexing generalized degenerate strings (resp. elastic-degenerate strings), showing that no combinatorial algorithm can answer queries in $\cO(n^{\cO(1)}m^{1-\varepsilon}+m)$ time (resp. $\cO(n^{\cO(1)}m^{2-\varepsilon})$) after polynomial preprocessing, unless the $k$-clique conjecture fails. We also showed that active-prefix queries cannot be answered in $\cO(m^{2-\varepsilon})$ time after preprocessing the string set and pattern in polynomial time, unless the OMv conjecture fails. We leave open the following questions:

\begin{itemize}
   \item Is the $\tcO(N\sqrt{m})$-time algorithm for generalized degenerate string matching optimal? Can we improve it to near-linear time, or prove a matching lower bound? Another natural target is to match the running time of the quantum algorithm of Equi et al.~\cite{equiQuantum2026}, namely $\tcO(\sqrt{Nmn})$ time.
   \item Could we extend our lower bounds to general algorithms, not only combinatorial ones?
   \item A symmetric taxonomy was considered for the graph variant of each type of variable string~\cite{asconeUnifying2024}, where adjacent segments in the text are connected by edges and where concatenations are allowed only along paths in the graph. In particular, the same gap was observed: it is not known whether one can solve pattern matching in \emph{founder graphs}, the graph variant of generalized degenerate strings, in subquadratic time. Can we adapt our techniques to this setting?
\end{itemize}

\section*{Acknowledgments}
I thank Itai Boneh and Paweł Gawrychowski for their helpful comments and suggestions.

\printbibliography
\end{document}